\documentclass[a4paper]{amsart}
\usepackage{amssymb}
\usepackage{latexsym}
\usepackage{tikz}
\usepackage{enumerate}
\usepackage{todonotes}
\usepackage{url}

\DeclareMathOperator{\AD}{AD}
\DeclareMathOperator{\Ad}{Ad}
\DeclareMathOperator{\ad}{ad}
\DeclareMathOperator{\Diff}{Diff}
\DeclareMathOperator{\SO}{SO}

\DeclareMathOperator{\SDiff}{SDiff}
\DeclareMathOperator{\trace}{trace}
\DeclareMathOperator{\Iso}{Iso}
\DeclareMathOperator{\GL}{GL}
\DeclareMathOperator{\Mat}{Mat}

\newtheorem{theorem}{Theorem}[section]

\newtheorem{proposition}{Proposition}[section]
\newtheorem{definition}{Definition}[section]

\newtheorem{remark}{Remark}[section]

\newcommand{\blangle}{\Big{\langle}}
\newcommand{\brangle}{\Big{\rangle}}
\newcommand{\ra}{\rightarrow}
\newcommand{\pder}[2]{\ensuremath{\frac{\partial #1}{\partial #2} } }
\newcommand{\e}{\ensuremath{ {\mathbf e} } }
\newcommand{\T}{\ensuremath{\mathcal{T}}}
\renewcommand{\S}{\ensuremath{\mathcal{S}}}

\title{Lagrangian mechanics on centered semi-direct products}
\author{Leonardo Colombo \& Henry O. Jacobs}
\date{25 June 2013}

\begin{document}

\maketitle

\begin{abstract}
    There exists two types of semi-direct products between a Lie group $G$ and a vector space $V$.
    The left semi-direct product, $G \ltimes V$, can be constructed when $G$ is equipped with a left action on $V$.
    Similarly, the right semi-direct product, $G \rtimes V$, can be constructed when $G$ is equipped with a right action on $V$.
    In this paper, we will construct a new type of semi-direct product, $G \Join V$, which can be seen as the `sum' of a right and left semi-direct product.
    We then parallel existing semi-direct product Euler-Poincar\'{e} theory.
    We find that the group multiplication, the Lie bracket, and the diamond operator can each be seen as a sum of the associated concepts in right and left semi-direct product theory.
    Finally, we conclude with a toy example and the group of $2$-jets of diffeomorphisms above a fixed point.
    This final example has potential use in the creation of particle methods for problems on diffeomorphism groups.
\end{abstract}

\section{Introduction}
  It is no secret that the use of symmetry and a preference for algebraic simplicity pervaded much (if not all) of Jerry's intellectual endeavours.
  Certainly one of these algebraic structures would be semi-direct products, which pepper his research in the form of rigid bodies, complex fluids, plasmas \cite{MaRaWe1984,MaRaWe1984B}, the KdV equation \cite{MaMiGePeRa1998}, and the heavy top \cite{HoMaRa}.
  
  In this paper we will investigate a new semi-direct product which is inspired by a careful analysis of the second order jet groupoid.  To begin, let $G$ be a Lie group and $V$ be a vector space on which $G$ acts by a left action.  Given these ingredients, we may form the Lie group $G \ltimes V$, which is isomorphic to $G \times V$ as a set, but equipped with the composition
\[
    (g,v) \cdot_{\ltimes} (h,w) = (g \cdot h, g \cdot w + v) \quad , \quad \forall (g,v), (h,w) \in G \ltimes V.
\]
A standard example of a system which evolves on a left semi-direct product is the heavy top, where $G = \SO(3)$ and $V = \mathbb{R}^3$.
In contrast, if $G$ acts on $V$ by a right action, we may form the right semi-direct product $G \rtimes V$ defined by the composition
\[
    (g,v) \cdot_{\rtimes} (h,w) = (g \cdot h , w + v \cdot h ).
\]
A standard example of a system whos configurations describe a right semi-direct product is a fluid with a vector-valued advected parameter \cite{HoMaRa}.
In any case, it seems natural to surmise that the composition law
\begin{align}
    (g,v) \cdot_{\Join} (h,w) = (g \cdot h , g \cdot w + v \cdot h) \label{eq:comp}
\end{align}
yields a new type of semi-direct product.  The first result of this article is that \eqref{eq:comp} is a valid composition law in some circumstances, and we call the corresponding Lie group a \emph{centered semi-direct product}.

  The second result is that the second order Taylor expansions (or second order jets) of diffeomorphisms over a fixed point form a centered semi-direct product.  The main motivation behind understanding this example is to allow us to develop particle-based methods for complex fluid simulation and image registration algorithms.

\subsection{Background}
  The semi-direct product is a standard tool used in the construction of new Lie groups and plays an interesting role in geometric mechanics when the normal subgroup is interpreted as an advected parameter.  A standard example is the modeling of the `heavy-top', wherein the the axis of rotation is described by $\mathbb{R}^3$ and is advected by the action of $\SO(3)$.  In other words, the configuration space for the heavy top can be described as the left semi-direct product $\SO(3) \ltimes \mathbb{R}^3$ \cite{HoMaRa}.  Another standard example is the modeling of liquid crystals, in which we  consider the right semi-direct product $\SDiff(M) \rtimes V$.  In this case, $\SDiff(M)$ is the set of volume-preserving diffeomorphisms of a volume manifold $M$, and $V$ is a vector space of maps from $M$ into some Lie algebra and $\SDiff(M)$ acts  on $V$ by pullback \cite{Holm2002,GayBalmaz2009}.  Of course, the tangent bundle of a Lie group, $TG$, is isomorphic to a left semi-direct product $G \ltimes \mathfrak{g}$ by left-trivializing the group structure of $TG$.  Additionally, $TG$ is isomorphic to a right semi-direct product $G \rtimes \mathfrak{g}$ when the group structure of $TG$ is right trivialized \cite[section 5.3]{MTA}.  Thus, we see that this method of constructing groups can be found in a number of instances.  In this article, we introduce a new type of semi-direct product which extends the existing semi-direct product theory.

  A motivating example will be a desire to understand the second order jet-groupoid of a manifold $M$ \cite[section 12]{KMS}.  As will be illustrated in section \ref{sec:jets}, an isotropy group of the second order jet groupoid exhibits a group structure which can be written as a centered semi-direct product.  A thorough understanding of the jet groupoid can be useful for the creation of new particle-based methods wherein the particles carry jet data in addition to position and velocity data.  One advantage of such a particle method is the possibility for a discrete form of Kelvin's circulation theorem \cite{JaRaDe2011}.  Building such particle methods can be useful in scenarios in which one desires to work with the material representation of a fluid.  For example, the free energy of liquid crystal is a function of the gradient of a director field advected by the fluid. Computing this advection requires the use of second order jet data and therefore a small portion of the material representation of the fluid is invoked \cite{Holm2002,GayBalmaz2009}.
  Additionally, the use of jet data can be useful in the realm of image registration algorithms in the field of medical imaging.  In particular, it is common to use the material representation of the EPDiff equations to implement the Large Deformation Diffeomorphic Metric Mapping (LDDMM) framework \cite{Beg2005,Bruveris2011}.  In particular, ``Landmark LDDMM'' discretizes the EPDiff equation using particle methods \cite{MuDe2010}.  A version of Landmark LDDMM wherein the particles can carry higher order jet data is described in \cite{Sommer2013}. Thus, keeping track of jet data may play a significant role in the construction of particle-based integrators for fluid modeling and medical imaging algorithms.

\subsection{Main Contributions}
In this paper, we accomplish a sequence of goals, each building upon the previous.  In particular:
\begin{enumerate}
    \item In section \ref{sec:CSD}, we define a new type of semi-direct product that we dub a \emph{centered semi-direct product}.
    \item In proposition \ref{prop:algebra}, we derive the Lie algebra of a centered semi-direct product and its associated structures.
    \item In section \ref{sec:EP}, we develop the Euler-Poincar\'e theory of centered semi-direct products in parallel with the existing theory of semi-direct product reduction \cite{HoMaRa}.
    \item In section \ref{sec:examples}, we describe the centered semi-direct product Euler-Poincar\'e equations for a few examples.  We present one toy example before presenting the theory for an isotropy group of the second order jet groupoid.
\end{enumerate}

Combined, these items allow for a computationally tractable algebraic understanding of second order jets and perhaps open the door to applications which were previously overlooked by geometric mechanicians.

\subsection{Acknowledgements}
We would like to thank Darryl D. Holm for providing the initial stimulus for this project.
The work of L.C has been supported by MICINN (Spain) Grant MTM2010-21186-C02-01, MTM 2011-15725-E, ICMAT Severo Ochoa Project SEV-2011-0087 and IRSES-project "Geomech-246981''.
L.C owes additional thanks to CSIC and the JAE program for a JAE-Pre grant.
The work of H.O.J. was supported by European Research Council Advanced Grant 267382 FCCA.

\subsection{A motivating example} \label{sec:jets}
  Let $\Diff(M)$ denote the diffeomorphisms group of a manifold $M$.  For a fixed $x \in M$ we may define the isotropy subgroup
  \[
    \Iso(x) = \{ \varphi \in \Diff(M)  \quad \vert \quad \varphi(x) = x \}.
  \]
  Let $\varphi \in \Iso(x)$ and note that $T_x \varphi$ is a linear automorphism of the vector-space $T_xM$. In particular:

\begin{proposition} \label{prop:jets}
  The functor ``$T_x$'' is a group homomorphism from $\Iso(x)$ to $\GL( T_x M)$.
\end{proposition}
\begin{proof}
    Clearly $\Iso(x)$ and $\GL(T_xM)$ are both Lie groups.  Let $\varphi, \psi \in \Iso(x)$.  Then $T_x \varphi \circ T_x \psi = T_x( \varphi \circ \psi)$.
\end{proof}

This observation has implications for computation for the following
reason: By definition, $T_x \varphi$ approximates $\varphi$ in a
neighborhood of $x \in M$. Thus, if one desired to model a continuum
with activity at $x$, then $T_x \varphi$ carries some of the crucial
data to do this task. In particular, this is computationally
tractable as the dimension of $\GL(T_x M)$ is equal to $(\dim M)^2$.


However, the group $\GL(n)$ only captures the linearization of a diffeomorphism.
If we desire to capture some of the nonlinearity then we might consider looking into the second jet of these diffeomorphisms.
We can do so by considering the functor $TT_x$.
Let $\varphi \in \Iso(x)$ so that $TT_x \varphi$ is a map from $T(T_xM)$ to $T(T_xM)$.
However, $T_xM$ is a vector-space so that $T(T_xM) \approx T_x M \times T_xM$.
The second component represents the vertical component and the isomorphism between $TT_xM$ and $T_x M \times T_xM$ is given by the vertical lift
\[
    v^{\uparrow}(v_1,v_2) =  \left. \frac{d}{d \epsilon} \right|_{\epsilon = 0} ( v_1 + \epsilon v_2).
\]
We can therefore represent $TT_x \varphi$ as $(T_x \varphi, A_{\varphi})$ where $A_{\varphi}: T_x M \times T_xM \to T_xM$ is the symmetric $(1,2)$ tensor
\begin{align}
    A_{ij}^{k} = \frac{\partial^2 \varphi^k}{\partial x_i \partial x_j }(x) \label{eq:12_tensor}
\end{align}
where $\varphi^k$ is the $k$th component of $\varphi$.  In other words, upon choosing a Riemannian metric to induce an coordinate system at $x$ we obtain the 1-1 correspondence
\[
    TT_x \varphi \leftrightarrow  (A_1,A_2)
\]
where $A_1 = \pder{\varphi^i}{x^j}$ and $A_2$ is given by \eqref{eq:12_tensor}.  If we denote the set of rank  $(1,2)$-tensors on $T_x M$ which are symmetric in the covariant indices by $\S^1_2(x)$, then this correspondence is given by a map
\[
    \Psi : \left. \mathcal{J}^2 \right|_{x}^{x}( \Diff(M) ) \to \GL( T_x M) \times \S^1_2(x)
\]
where $\left. \mathcal{J}^2 \right|_{x}^{x}( \Diff(M) )$ is the group of second order taylor expansions about $x$ of diffeomorphisms which send $x$ to itself (these are called second order jets).  This allows us to write the Lie group structure of $ \left. \mathcal{J}^2 \right|_{x}^{x} ( \Diff(M) )$ as a type of semi-direct product. In particular:
\begin{proposition} \label{prop:2jets}
    If we represent $TT_x \varphi$ and $TT_x \psi$ as $(A_1,A_2)$ and $(B_1,B_2)$ where $A_1 = T_x \varphi, B_2 = T_x \psi, A_2 = \frac{ \partial^2 \varphi^k}{ \partial x^i \partial x^j}$,  and $B_2 = \frac{ \partial^2 \psi^k}{\partial x^i \partial x^j }$, then $TT_x \varphi \circ TT_x \psi \equiv TT_x( \varphi \circ \psi)$ is given by the composition
    \[
        (A_1,A_2) \circ (B_1,B_2) = (A_1 \circ B_1 , A_1 \circ B_2 + A_2 \circ (B_1 \times B_1) ).
    \]
\end{proposition}
\begin{proof}
    We find that
    \[
        \pder{}{x_i} ( \varphi^k \circ \psi) = \pder{\varphi^k}{x_l} \cdot \pder{\psi^l}{x_i} \circ \psi
    \]
    and the second derivative is
    \begin{align*}
        \pder{}{x_j} \pder{}{x_i} (\varphi^k \circ \psi) &= \pder{}{x_j} \left( \pder{\varphi^k}{x_l} \cdot \pder{\psi^l}{x_i} \circ \psi \right) \\
            &= \left( \frac{\partial^2 \varphi^k}{ \partial  x_l \partial x_m } \pder{\psi^l}{x_i} \pder{\psi^m}{x_j} + \pder{\varphi^k}{x_l} \frac{ \partial^2 \psi^l}{\partial x_i \partial x_j} \right) \circ \psi.
    \end{align*}
    Noting that $\psi(x) = x$ we can set
    \begin{align*}
        A_1 = \left. \pder{\varphi^k}{x_l} \right|_{x} , \quad A_2 = \left. \frac{\partial^2 \varphi^k}{\partial x_i \partial x_j} \right|_{x} \\
        B_1 = \left. \pder{\psi^k}{x_l} \right|_{x} , \quad B_2 = \left. \frac{\partial^2 \psi^k}{\partial x_i \partial x_j} \right|_{x}
    \end{align*}
    and rewrite the equations in the form
    \begin{align*}
        \pder{}{x_i} ( \varphi^k \circ \psi) &= A_1 \cdot B_1\\
        \pder{}{x_j} \pder{}{x_i} (\varphi^k \circ \psi) &= A_1 \cdot B_2 + A_2 \circ (B_1 \times B_1).
    \end{align*}
    Therefore, if we define the composition
    \[
        (A_1, A_2) \cdot (B_1, B_2) := (A_1 \cdot B_1 , A_1 \cdot B_2 + A_2 \circ (B_1 \times B_1) )
    \]
    on the manifold $\GL(T_x M) \times \S^1_2$, then $\Psi: \left. \mathcal{J}^2 \right|_{x}^{x}( \Diff(M) ) \to \GL(T_x M ) \times \S^1_2$ is a Lie group isomorphism by construction.
\end{proof}

We see that the composition law of Proposition \ref{prop:2jets} is
of the form described in equation \eqref{eq:comp}.  In this paper,
we will condense the composition law for second order jets to the algebraic
level and study \eqref{eq:comp} in the abstract Lie group setting.
Of course, one would naturally like to consider diffeomorphisms
which are not contained in $\Iso(x)$.  However, this extension
brings us into the realm of Lie groupoid theory and will need to be
addressed in future work.

\section{A centered semi-direct product theory} \label{sec:CSD}
In this section, we will discover a new type of semi-direct product.  We will outline the necessary ingredients for the construction of such a Lie group and we will derive the corresponding structures on the Lie algebra.

\subsection{Preliminary material on Lie groups}
 Let $G$ be a Lie group with identity $e\in G$ and Lie algebra $\mathfrak{g}$.
  In this subsection we will establish notation and recall relevant notions related to Lie groups and Lie
algebras.


\subsubsection{Group actions:}
Let $V$ be a vector space. A \textit{left
action} of $G$ on $V$ is a smooth map $\rho_{L}:G\times V\ra V$ for
which:
\[
    \rho_{L}(e,v)=v \text{ and } \rho_{L}(g,\rho_{L}(h,v))=\rho_{L}(gh,v) \quad , \quad \forall g,h\in G , \forall v\in V.
\]
  As using the symbol `$\rho_L$' can become cumbersome and since we will only need
 a one left Lie group action in a given context, we will opt to use the notation
$g\cdot v := \rho_{L}(g,v).$ Finally, the \textit{induced infinitesimal left
action} of $\mathfrak{g}$ on $V$ is
\[
\xi\cdot v := \frac{d}{d\epsilon}\Big{|}_{\epsilon=0} \exp( \epsilon \cdot \xi ) \cdot v \quad , \quad \forall \xi \in \mathfrak{g} , v \in V.
\]
Similarly, a \textit{right action} of $G$ on
$V$ is the smooth map $\rho_{R}:V\times G\ra V$ for which:
\[
    \rho_{R}(v,e)=v \text{ and } \rho_{L}(\rho_{L}(v,g),h)=\rho_{L}(v,gh) \quad , \quad \forall g,h\in G , \forall v\in V.
\]
  Again, we will primarily use the notation
$v\cdot g:= \rho_{R}(v,g)$ for right actions. The \textit{induced infinitesimal right
action} of $\mathfrak{g}$ on $V$ is given by
\[
v\cdot\xi =\frac{d}{d\epsilon}\Big{|}_{\epsilon=0}v\cdot \exp( \epsilon \cdot \xi) \quad, \quad \forall \xi \in \mathfrak{g}, v \in V
\]

Lastly, we say that the left action and the right action \emph{commute} if
 \[
    (g \cdot v) \cdot h = g \cdot (v \cdot h)
\]
for any $g,h \in G$ and $v \in V$.

\subsubsection{Adjoint and coadjoint operators:}

  In this section we will recall the ``$\AD, \Ad, \ad$''-notation used in \cite{Holm_GM}. For $g\in G$ we define the \textit{inner automorphism} $\AD:G \times G\ra
G$ as $\AD(g,h) \equiv \AD_{g}(h)= g h g^{-1}$.  Differentiating $\AD$ with respect to the second argument along curves through the identity produces the \textit{Adjoint representation} of $G$ on $\mathfrak{g}$ denoted $\Ad:G\times\mathfrak{g}\ra\mathfrak{g}$ and given by
\[
    \Ad_{g}(\eta)= \left. \frac{d}{d\epsilon}\right|_{\epsilon=0} \left( \AD_g( \exp(\epsilon \eta) ) \right) = g \cdot \eta \cdot g^{-1},
\]
 for $g\in G$ and $\xi\in\mathfrak{g}$.  Differentiating $\Ad$ with respect to the first argument along curves through the identity produces the \textit{adjoint} operator $\ad:\mathfrak{g}\times\mathfrak{g}\ra\mathfrak{g}$ given by
\[
    \ad_{\xi}(\eta)= \left. \frac{d}{d\epsilon} \right|_{\epsilon = 0} ( \Ad_{ \exp( \epsilon \xi) }(\eta) ) = \xi \cdot \eta - \eta \cdot \xi.
\]
The $\ad$-map is an alternative notation for the Lie bracket of $\mathfrak{g}$ in the sense that
\[
    \ad(\xi,\eta) \equiv \ad_{\xi}(\eta) \equiv [\xi,\eta].
\]

For each $\xi \in \mathfrak{g}$ the map $\ad_{\xi} : \mathfrak{g}
\to \mathfrak{g}$ is linear and therefore has a formal dual
$\ad_\xi^*: \mathfrak{g}^* \to \mathfrak{g}^*$ which we call the
\emph{coadjoint operator}. Explicitly, $\ad_{\xi}^*$ is defined by
the relation
 \begin{align}
    \langle \ad_{\xi}^{*} (\mu) ,\eta \rangle= \langle \mu,\ad_{\xi}(\eta)\rangle \label{eq:ad_star}
\end{align}
for each $\eta \in \mathfrak{g}$ and $\mu \in \mathfrak{g}^*$.

\subsection{Centered semi-direct products} \label{subsec:CSD}
  In this subsection, we will construct a semi-direct product which can be thought of as a `sum' of a right semi-direct product and a left semi-direct product.
\begin{proposition}
Let $G$ be a Lie group which acts on a vector-space $V$ via left and right group actions.  Then, the product $G\times V$ with the composition law
\begin{equation}\label{cs-dp}
(g_{1},v_{1})\cdot(g_{2},v_{2}):=(g_{1}g_{2},g_{1}\cdot v_{2}+v_{1}\cdot g_{2})
\end{equation} is a Lie group if and only if the left and right actions of $G$ commute.
\end{proposition}

\begin{proof}
    It is clear that $G \times V$ is a smooth manifold and that the composition law \eqref{cs-dp} is a smooth map.  We must prove that this composition makes $G \times V$ a group.
    \begin{itemize}
        \item That the composition map \eqref{cs-dp} produces another element of $G \times V$ can be observed directly.  Thus `closure' is satisfied.
        \item The identity element is given by $(e,0) \in G \times V$ where $e\in G$ is the identity of $G$.
        \item The inverse element of an arbitrary $(g,v)\in G\times V$ is $(g^{-1},-g^{-1}vg^{-1})$ where $g^{-1}$ is the inverse of $g \in G.$
        \item Given three elements of $G \times V$ we find
        \begin{align*}
            (g_1,v_1)\cdot \left( (g_2,v_2)\cdot(g_3,v_3) \right) = (g_1,v_1) \cdot (g_2 g_3 , g_2 \cdot v_3 + v_2 \cdot g_3 ) \\
            = \left( g_1 g_2 g_3,g_1\cdot(g_2\cdot v_3+v_2 \cdot g_3)+v_1\cdot(g_2g_3) \right) \\
            = \left( (g_1 g_2) g_3, (g_1g_2)\cdot v_3+g_1\cdot(v_2\cdot g_3)+(v_1\cdot g_2)\cdot g_3 \right).
        \end{align*}
        By the commutativity of the group actions we may equate the above line with:
        \begin{eqnarray*}
            &=& ((g_1g_2)g_3, (g_1g_2)\cdot v_3+(g_1\cdot v_2)\cdot g_3+(v_1\cdot g_2)\cdot g_3)\\
            &=& ((g_1g_2)g_3, (g_1g_2)\cdot v_3+(g_1\cdot v_2+v_1\cdot g_2)\cdot g_3)\\
            &=& ((g_1g_2), g_1\cdot v_2+v_1\cdot g_2)\cdot(g_3,v_3)\\
            &=& ((g_1,v_1)\cdot (g_2,v_2))\cdot(g_3,v_3).
        \end{eqnarray*}
        Thus, the associative property is satisfied.
    \end{itemize}
    Moreover, all maps in sight including the inverse map are smooth.  In conclusion we see that $G \times V$ with the composition \eqref{cs-dp} defines a Lie group.  Moreover, if the left and right actions of $G$ on $V$ do \emph{not} commute, then we can observe that associativity is violated.
\end{proof}

\begin{definition}\label{def-bowtie-Lie group}

Given commuting left and right representations of a group $G$ on a vector space $V$, the Lie group $G\times V$ with the composition \eqref{cs-dp} is
denoted $G \Join  V$ and called the \emph{centered semi-direct product} of $G$ and $V.$
\end{definition}

  It customary to denote the left semi-direct product using the symbol $\ltimes$ and the right semi-direct product via the symbol $\rtimes$.  We justify our use of the symbol $\Join $ in that the concept of centered semi-direct product is merely a `sum' of a left and a right semi-direct product.  The formula  $\Join  = \rtimes + \ltimes$ can be used as a heuristic throughout the paper.  In particular, this heuristic applies to the Lie algebra.

\begin{proposition} \label{prop:algebra}
Let $G \Join  V$ be a centered-semi direct product Lie group.  The Lie algebra $\mathfrak{g} \Join V$ is given by the set $\mathfrak{g} \times V$ with the Lie bracket
\begin{equation}\label{bracket}
\left[(\xi_1, v_1),(\xi_2, v_2)\right]_{\Join }
=\left([\xi_1,\xi_2]_{\mathfrak{g}},
(\xi_1\cdot v_2+v_1\cdot\xi_2)-(\xi_2\cdot v_1+v_2\cdot \xi_1)\right),
\end{equation} for $\xi_1,\xi_2\in\mathfrak{g},$ $v_1,v_2\in V$.
\end{proposition}

\begin{proof}
  Firstly, it is simple to verify that the tangent space at the identity, $(e,0) \in G \times V$, is
$\mathfrak{g}\times V$.  To derive the Lie bracket, we will derive the the $\ad$-map via the $\Ad$ and $\AD$-maps.
For $(g,v), (h,w)\in G\Join  V$ we find
\begin{eqnarray*}
\AD_{(g,h)}(h,w)&=&(gh,v\cdot h + g\cdot w)\cdot(g^{-1},-g^{-1}\cdot v\cdot g^{-1})\\
&=&(\AD_{g}(h), v\cdot hg^{-1}+g\cdot w\cdot g^{-1}-\AD_{g}(h)\cdot v\cdot g^{-1}).
\end{eqnarray*}

  If we substitute $(h,w)$ with the $\epsilon$-dependent curve $( \exp( \epsilon \cdot \xi_2) , \epsilon \cdot v_1)$ we can calculate the \textit{adjoint operator}, $\Ad:(G\Join
V)\times(\mathfrak{g}\Join  V)\rightarrow \mathfrak{g}\Join  V.$  Given by

\begin{eqnarray*}
\Ad_{(g,v)}(\xi_2,v_2)&=& \left. \frac{d}{d \epsilon }\right|_{\epsilon = 0}\AD_{(g,v)}(\exp(\epsilon \cdot \xi_1) ,  \epsilon \cdot v_1)\\
&=&(\Ad_{g}(\xi_2); v\cdot \xi_2 g^{-1}+g\cdot v_2\cdot g^{-1}-\Ad_{g}(\xi_2)\cdot v\cdot g^{-1}).
\end{eqnarray*}
If we substitute $(g,v)$ with the $t$-dependent curve $( \exp( t \xi_1) , t v_2)$ we can differentiate with respect to $t$ to produce the
adjoint operator $\ad:(\mathfrak{g}\Join  V)\times (\mathfrak{g}\Join  V)\rightarrow \mathfrak{g}\Join  V$.  Specifically, the adjoint operator is given by
\begin{eqnarray*}
\ad_{(\xi_1,v_1)}(\xi_2,v_2)&=&\frac{d}{dt}\Big{|}_{t=0}(\Ad_{(\exp(t \cdot \xi_1),t \cdot v_1)}(\xi_2,v_2))\\
&=&\frac{d}{dt}\Big{|}_{t=0}(g\xi_2g^{-1},v\cdot\xi_2g^{-1}-g\xi_2g^{-1}\cdot v\cdot g^{-1}+ g\cdot v_2\cdot g^{-1})\\
&=&( \ad_{\xi_1}( \xi_2) , \xi_1 \cdot v_2 + v_1 \cdot \xi_2 - \xi_2 \cdot v_1 - v_2 \cdot \xi_1)\\
&=&([\xi_1,\xi_2]_{\mathfrak{g}}, (\xi_1\cdot v_2+v_1\cdot \xi_2) - (\xi_2\cdot v_1 + v_2 \cdot \xi_1 ) ).
\end{eqnarray*}

Noting that the $\ad$-map is merely an alternative notation for the Lie bracket completes the proof.
\end{proof}


We complete this section by defining operations designed to express
interaction terms between momenta in $V$ and momenta in $G$ in
mechanical systems.
\begin{definition}
The \emph{heart operator} $\heartsuit : \mathfrak{g}\times V^{*}\rightarrow V^{*}$ is defined by
\begin{equation}\label{triangle}
\langle \xi \heartsuit \alpha, v\rangle_{V}:=\langle \alpha,\xi\cdot v-v\cdot\xi\rangle_{V}.
\end{equation}
The \emph{diamond operator}, $\diamondsuit:V\times V^{*}\ra\mathfrak{g}^{*}$, is defined as
\begin{align}
    \langle v \diamondsuit \alpha,\xi\rangle_{\mathfrak{g}}:=\langle\alpha, v\cdot\xi-\xi\cdot v\rangle_{V}.
\end{align}
\end{definition}

  The diamond operator can be seen as the sum of a diamond operator of a left semi-direct product and that of a right semi-direct product ~\cite{HoMaRa}.  If we view $G \bowtie V$ as a Lie group and take the corresponding Line variations then the heart operator and diamond operator comes into play.  However, if we restrict the variations so that $V$ acts as an advected parameter, only the diamond operator is present.  We will elaborate on both these options in the next section.

\section{Euler-Poincar\'e theory} \label{sec:EP}
  The Euler-Lagrange equations on a Lie group, $\tilde{G}$, can be expressed by a vector field over $T\tilde{G}$.
  If the Lagrangian is $\tilde{G}$-invariant then the equations of motion are $\tilde{G}$-invariant as well and the evolution equations can be reduced.
  While the unreduced system evolves by the \emph{Euler-Lagrange} equations on $T \tilde{G}$, the reduced dynamics evolve on the quotient $T\tilde{G} / \tilde{G}$.
  However, $T\tilde{G} / \tilde{G}$ is just an alternative description of the Lie algebra $\tilde{\mathfrak{g}}$ and so the reduced equations of motion can be described on $\tilde{\mathfrak{g}}$ where we call them the \emph{Euler-Poincar\'{e} equations.}  This reduction procedure is summarized by the commutative diagram:

\begin{center}
\begin{tikzpicture}[node distance=2.5cm, auto]
  \node (GL) {$T\tilde{G}$};
  \node (GR) [right of=GL]{$T\tilde{G}$};
  \node (gL) [below of=GL] {$\tilde{\mathfrak{g}}$};
  \node (gR) [right of=gL]{$\tilde{\mathfrak{g}}$};
  \draw[->] (GL) to node {flow by `EL'} (GR);
  \draw[->] (gL) to node [swap] {flow by `EP'} (gR);
  \draw[->] (GL) to node [swap] {$/ \tilde{G}$} (gL);
  \draw[->] (GR) to node {$ / \tilde{G}$} (gR);
\end{tikzpicture}
\end{center}

  To be even more specific.  A Lagrangian $L: T\tilde{G} \to \mathbb{R}$ is said to be \emph{(right) $\tilde{G}$-invariant} if
  \[
    L( (\tilde{g}, \dot{\tilde{g}}) \cdot h) = L(\tilde{g} , \dot{\tilde{g}})
  \]
  for all $h \in \tilde{G}$.  If $L$ is $\tilde{G}$-invariant, then $L$ is uniquely specified by its restriction $\ell = \left. L \right|_{ \tilde{ \mathfrak{g}}} : \tilde{\mathfrak{g}} \to \mathbb{R}$.  The Euler-Poincar\'e theorem states that the Euler-Lagrange equations
  \[
    \frac{d}{dt} \left( \frac{ \delta L}{ \delta \dot{\tilde{g}}} \right) - \frac{ \delta L}{ \delta \tilde{g}} = 0
  \]
  on $T\tilde{G}$ are equivalent to the Euler-Poincar\'{e} equations and reconstruction formula
  \[
    \frac{d}{dt} \left( \frac{ \delta  \ell}{ \delta \tilde{\xi} } \right) = - \ad_{\tilde{\xi}}^* \left( \frac{ \delta \ell}{ \delta  \tilde{\xi}} \right) \quad , \quad \tilde{\xi} := \dot{\tilde{g}} \cdot \tilde{g}^{-1}.
  \]
  A review of Euler-Poincar\'{e} reduction is given in ~\cite[Ch 13]{MandS} while a specialization to the case of semidirect products with advected parameters is described in ~\cite{HoMaRa}.  In this section we will specialize the Euler-Poincar\'{e} theorem to the case of centered semi-direct products by setting $\tilde{G} = G \Join  V$.

To begin let us compute how variations of curves in the group induce variations on the trivializations of the velocities to the Lie algebra.  Studying such variations will allow us to transfer the variational principles on the group  to variational principles on the Lie algebra.

\begin{proposition} \label{prop:variations}
Let $G \Join V$ be a centered semi-direct product and consider a curve  $(g,v)(t) \in G \Join  V$.
Let $(\xi_{g}(t),\xi_{v}(t)):=(\dot{g}(t),\dot{v}(t))\cdot(g(t),v(t))^{-1}\in\mathfrak{g}\Join V$ be the right trivialization of $(\dot{g},\dot{v})(t)$.
An arbitrary variation of $(g,v)(t)$ is given by
\[
    (\delta g , \delta v)(t) = (\eta_g , \eta_v)(t) \cdot (g,v)(t) \in T_{(g,v)(t)} (G \Join  V),
\]
where $(\eta_g,\eta_v)(t) \in \mathfrak{g} \Join  V$.
Given such a variation, the induced variation on $(\xi_g,\xi_v)$ is given by
\begin{align}\label{variationsright}
(\delta\xi_g, \delta\xi_v) &= (\dot{\eta}_{g}-\ad_{\xi_{g}}\eta_{g}, \dot{\eta}_{v}+ (\eta_g \xi_v + \eta_v \xi_g) - (\xi_{g}\eta_{v}+\xi_{v}\eta_{g}))  \\
    &= \frac{d}{dt} (\eta_v,\eta_v) - [ (\xi_g,\xi_v) , (\eta_g, \eta_v) ]_\Join. \nonumber
\end{align}
\end{proposition}

\begin{proof}
    For any Lie group, $\tilde{G}$, and any curve $\tilde{g}(t) \in \tilde{G}$, the variation of $\tilde{\xi}(t) := \dot{\tilde{g}}(t) \cdot \tilde{g}^{-1}(t)$ induced by the variation $\delta \tilde{g}(t)  = \tilde{\eta}(t) \cdot \tilde{g}(t)$ is $\delta \tilde{\xi} = \dot{\tilde{\eta}} - [ \tilde{\xi} , \tilde{\eta} ]$.  For matrix groups see \cite[Theorem 13.5.3]{MandS} and \cite{BlKrMaRa} for the general case.  If we set $\tilde{G} = G \Join  V$ and use the bracket derived in Proposition  \ref{prop:algebra} then the theorem follows.
\end{proof}

Now that we understand the relationship between variations of curves in $G \Join  V$ and the induced variations in $\mathfrak{g} \Join  V$ we can state the Euler-Poincar\'e theorem for centered semi-direct products.

\begin{theorem}\label{thm:ep}
Let $L: T(G \Join  V) \to \mathbb{R}$ be (right) $G \Join
V$-invariant, and let $\ell: \mathfrak{g} \Join  V \to \mathbb{R}$
be its reduced Lagrangian.  Let $(g,v)(t) \in G \Join  V$ and denote
the right trivialized velocity by $(\xi_g, \xi_v)(t) := (\dot{g} ,
\dot{v})(t) \cdot (g,v)(t)^{-1}$.  Then the following statements are
equivalent:
\begin{enumerate}[(i)]
\item Hamilton's principle holds.  That is,
\begin{equation}\label{action1}\delta\int_{t_0}^{t_1}L(g(t),\dot{g}(t),v(t))dt=0\end{equation} for
variations of $(g,v)(t)$ with fixed endpoints.

\item $(g, v)(t)$ satisfies the Euler-Lagrange equations for $L$.

\item The constrained variational principle
\begin{equation}\label{action2}
\delta\int_{t_0}^{t_1} \ell (\xi_{g}(t),\xi_v(t))dt=0\end{equation} holds on
$\mathfrak{g}\times V$ for variations of the form

\begin{equation}\label{eq:variations2L}
(\delta\xi_g, \delta\xi_v)=(\dot{\eta}_{g}-\ad_{\xi_{g}}\eta_{g}, \dot{\eta}_{v}+\eta_g\xi_v-\xi_{v}\eta_{g}+\eta_{v}\xi_{g}-\xi_{g}\eta_{v}),
\end{equation} where $(\eta_{g}, \eta_{v})(t)$ is an arbitrary curve in $\mathfrak{g}\Join  V$ which vanishes at the endpoints.

\item The Euler-Poincar\'e equations
\begin{eqnarray*}\label{EPeq1}
\frac{d}{dt}\left(\frac{\delta \ell}{\delta\xi_{g}}\right) + \ad_{\xi_g}^{*}\left(\frac{\delta \ell}{\delta\xi_g}\right) +\xi_{v} \diamondsuit \frac{\delta \ell}{\delta \xi_v}&=&0,\\
\frac{d}{dt}\left(\frac{\delta \ell}{\delta\xi_{v}}\right) + \xi_{g} \heartsuit \frac{\delta \ell}{\delta\xi_v}&=&0\nonumber
\end{eqnarray*}
hold on $\mathfrak{g}\Join  V$.
\end{enumerate}
\end{theorem}
\begin{proof}
 The equivalence {\it (i)} and {\it (ii) } holds for any
configuration manifold and so, in particular it holds in this case.

Next we show the equivalence {\it (iii)} and {\it (iv)}. We
compute the variations of the action integral to be
\begin{align*}
\delta \int_{t_0}^{t_1}\ell(\xi_{g}(t),\xi_v(t))dt =& \int_{t_0}^{t_1}\Big{\langle}\frac{\delta \ell}{\delta\xi_{g}},\delta\xi_g\Big{\rangle}+\Big{\langle}\frac{\delta \ell}{\delta\xi_{v}},\delta \xi_v\Big{\rangle}dt\\
=&\int_{t_0}^{t_1}\Big{\langle}\frac{\delta \ell}{\delta\xi_{g}},\dot{\eta}_{g}-\ad_{\xi_{g}}\eta_{g}\Big{\rangle}+\Big{\langle}\frac{\delta \ell}{\delta \xi_{v}},\dot{\eta}_{v}+\eta_{g}\xi_{v}-\xi_{v}\eta_{g}+\eta_{v}\xi_{g}-\xi_{g}\eta_{v}\Big{\rangle}dt\\
\intertext{ and applying integration by parts and equation \eqref{eq:ad_star} we find}\\
=&\int_{t_{0}}^{t_1}\blangle-\frac{d}{dt}\left(\frac{\delta \ell}{\delta \xi_{g}}\right) -\ad_{\xi_{g}}^{*}\left(\frac{\delta \ell}{\delta \xi_{g}}\right) , \eta_{g}\brangle
+\blangle-\frac{d}{dt}\frac{\delta\ell}{\delta\xi_{v}} , \eta_{v}\brangle\\
&+\blangle\frac{\delta \ell}{\delta \xi_{v}} , \eta_{g}\xi_{v} - \xi_{v}\eta_g\brangle
   +\blangle\frac{\delta \ell}{\delta\xi_{v}} , \eta_{v}\xi_{g} - \xi_g\eta_v\brangle dt \\
&+ \blangle\frac{\delta\ell}{\delta\xi_{g}} , \eta_{g}\brangle\Big{|}_{t_0}^{t_1}+\blangle\frac{\delta \ell}{\delta \xi_{v}} , \eta_{v}\brangle\Big{|}_{t_0}^{t_1}\\
=&\int_{t_0}^{t_1}\Big{\langle}-\frac{d}{dt}\left(\frac{\delta l}{\delta\xi_g}\right)-\ad_{\xi_g}^{*}\left(\frac{\delta \ell}{\delta\xi_g}\right)-\left(\xi_{v}\diamondsuit \frac{\delta \ell}{\delta \xi_{v}}\right), \eta_g\Big{\rangle}\\
&+\Big{\langle}-\frac{d}{dt}\left(\frac{\delta \ell}{\delta\xi_v}\right)- \xi_{g}\heartsuit \frac{\delta \ell}{\delta\xi_v}, \eta_v\Big{\rangle}dt.
\end{align*}
By noting that $(\eta_g,\eta_v)(t)$ is arbitrary on the interior of the integration domain, the result follows.

Finally, we show that {\it (i)} and {\it(iii)} are equivalent. The
$G-$invariance of $L$ implies that the
integrands in \eqref{action1} and \eqref{action2} are equal.
However, by Proposition \ref{prop:variations} all the variations of $(g,v)(t)$ with fixed endpoints induce, and are induced by,
variations $(\delta\xi_g,\delta\xi_v)(t)\in\mathfrak{g}\Join  V$
of the form given in equation \eqref{eq:variations2L}.
Conversely if {\it (i)} holds with respect to arbitrary variations $(\delta g, \delta v)$, we define
\[
    (\eta_{g},\eta_{v})(t) =  (\delta g , \delta v)  \cdot (g,v)^{-1},
\]
to produce the variation of $(\xi_g, \xi_v)$ given in equation \eqref{eq:variations2L}.
\end{proof}

\begin{remark}
There is a left invariant version of theorem \eqref{thm:ep} in which
$(\xi_{g} , \xi_v) :=(g,v)^{-1} \cdot (\dot{g},\dot{v})$ and $L$ is left $G \Join  V$-invariant. In
this case the Euler-Poincar\'e equations take the form
\begin{align*}
\frac{d}{dt}\left(\frac{\delta \ell}{\delta\xi_{g}}\right) - \ad_{\xi_g}^{*}\left(\frac{\delta \ell}{\delta\xi_g}\right) - \xi_{v} \diamondsuit \frac{\delta \ell}{\delta \xi_v}&=0,\\
\frac{d}{dt}\left(\frac{\delta \ell}{\delta\xi_{v}}\right) - \xi_{g} \heartsuit \frac{\delta \ell}{\delta\xi_v}&=0.
\end{align*}
\end{remark}

\begin{remark}
There is a version of semi-direct product mechanics wherein the vector-space $V$ is a set of \emph{advected parameters} as in \cite{HoMaRa}.  In this case we impose the holonomic constraint
\[
    \dot{v} = \dot{g} \cdot v + v \cdot \dot{g}
\]
and the set of admissible variations in $\mathfrak{g} \Join V$ become
\begin{align*}
    \delta \xi_g = \dot{\eta}_g - [\xi_g , \eta_g] \quad , \quad \delta v = \eta_g \cdot v + v \cdot \eta_g.
\end{align*}
If we do this, the $\heartsuit$-term is removed and $\frac{ \delta \ell}{\delta v }$ equation is replaced with a holonomic constraint.  In particular we find that
\begin{align*}
\frac{d}{dt}\left(\frac{\delta \ell}{\delta\xi_{g}}\right) \pm \ad_{\xi_g}^{*}\left(\frac{\delta \ell}{\delta\xi_g}\right) \pm \xi_{v} \diamondsuit \frac{\delta \ell}{\delta \xi_v}&=0 \\
\frac{d v}{dt} = \xi_g \cdot v + v \cdot \xi_g.
\end{align*}
where we use a plus sign for right trivialization and a minus sign for left trivialization.
\end{remark}

\section{Examples} \label{sec:examples}
In this section we will present two examples of Euler-Poincar\'{e} equations on centered semidirect products.  This first is a toy example designed to illustrate how computations of the diamond and heart operators can be done in practice.  The second example is concerns second order jets as described in subsection \ref{sec:jets}.

\subsection{A toy example}
  Consider the group $\GL(n)$ and let $\Mat(n)$ denote the vector space of $n \times n$ real matrices.  Noting that $\GL(n)$ acts on $\Mat(n)$ by left and right multiplication, we can define the composition law on the Lie group $GL(n) \Join  \Mat(n)$  by:
  \[
    (A,v)\cdot (B,w)=(AB,Aw+vB).
\]
Moreover, we can identify $\mathfrak{gl}^*(n)$ with
$\mathfrak{gl}(n)$ and $\Mat(n)^*$ with $\Mat(n)$ by the matrix
trace pairing $\langle A , B \rangle=\trace( A^{T}  B)$.  This
allows us to calculate the heart operator
$\heartsuit:\mathfrak{gl}(n)\times \Mat(n)^* \to \Mat(n)$ as

\begin{eqnarray*}
\langle A\heartsuit w,v\rangle&=&\langle w,A\cdot v-v \cdot A\rangle\\
&=&\trace \left(w^{T}(A\cdot v-v \cdot A)\right)\\
&=&\trace \left(w^{T}\cdot (A\cdot v)-w^{T}(v\cdot A)\right)\\
&=&\trace \left((w^{T}\cdot A)v-(A\cdot w^{T})\cdot v\right)\\
&=&\trace \left((w^{T}\cdot A-A\cdot w^{T})\cdot v\right)\\
&=&\trace \left((A^{T}w-w\cdot A^{T})^{T}\cdot v\right)\\
&=&\langle A^{T}w-w A^{T}, v\rangle
\end{eqnarray*}

Therefore,
\[
  A\heartsuit w = A^{T}w-w A^{T}.
\]

By a similar calculation, diamond operator is found to be
\[
    v \diamondsuit w = v^T w - w v^T,
\]
and the coadjoint action on $\GL(n)$ is given by
\[
    \ad_{A}^{*}(\alpha_A )= A^{T} \cdot \alpha_A-\alpha_A \cdot A^{T}.
\]

Now, we have all the ingredients to write the Euler-Poincar\'e
equations. Given a reduced Lagrangian $\ell: \mathfrak{gl}(n) \Join \Mat(n) \to \mathbb{R}$ we may denote the reduced momenta by
\[
    \mu = \frac{ \delta \ell}{\delta \xi} , \quad \gamma = \frac{ \delta \ell}{\delta v}.
\]
where $(\xi, v) \in \mathfrak{gl}(n) \Join \Mat(n)$.  The Euler-Poincar\'{e} equations can be written as
\begin{eqnarray*}
\dot{\mu}&=&(\xi^{T}\mu-\mu\xi^{T})+v^{T}\gamma-\gamma v^{T}\\
\dot{\gamma}&=&\xi^{T}\gamma-\gamma \xi^{T}.
\end{eqnarray*}

\subsection{An isotropy group of a second order jet groupoid}
 In proposition \ref{prop:jets} we illustrated how the second order jets of diffeomorphisms of the stabilizer group of a point $x \in M$ is identifiable with a centered semidirect product.
 In particular, if $\dim(M) = n$ we can consider the group $\GL(n) \Join \S^1_2$, where $\S^1_2$ is the set of $(1,2)$-tensors which are symmetric in the covariant indices.
 For the moment we shall consider the larger space of all $(1,2)$-tensors denoted $\T^1_2$.
 If we let $\e_1, \dots, \e_n \in \mathbb{R}^n$ be a basis with dual basis $\e^1, \dots , \e^n \in (\mathbb{R}^n)^*$ we can write an arbitrary element of $\T^1_2$ as
 \[
    T =  T^i_{jk} \e_i \otimes \e^j \otimes \e^k .
 \]
 The left action of $\GL(n)$ on $\T^1_2$ is
 \[
    g \cdot T := T^{i}_{jk} (g \cdot \e_i) \otimes \e^j \otimes \e^k \equiv T^{i}_{jk} g^{l}_i \e_l \otimes  \e^j \otimes \e^k
 \]
 while the right action is
 \[
    T \cdot g := T^{i}_{jk} \e_i \otimes ( g^T \cdot \e^j ) \otimes (g^T \cdot \e^k).
 \]
 Clearly these actions commute, and so we may form the centered semidirect product Lie group $\GL(n) \Join \T^1_2$.

 Let us now focus on the Lie algebra.  The Lie algebra $\mathfrak{gl}(n)$ is equivalent to $\T^1_1$ and the Lie bracket is then given in the bases $\e_i \otimes \e^j$ by
 \[
    [\xi , \eta] = (\xi^i_k \eta^k_j - \eta^i_k \xi^k_j) \e_i \otimes \e^j,
 \]
 where $\xi = \xi^i_j \e_i \otimes \e^j$ and $\eta = \eta^i_j \e_i \otimes \e^j$.
 We can use the dual basis $\e^i \otimes \e_j$ to see that the coadjoint action of $\xi$ on $\mu = \mu_i^j \e^i \otimes \e_j$ is given by
 \[
    \ad_\xi^* \mu = (\mu^j_k \xi^k_i - \mu^k_i \xi^j_k) \e^i \otimes \e_j.
 \]
  By differentiation we see that the infinitesimal left and right actions of $\mathfrak{gl}(n)$ on $\T^1_2$ are given by
 \begin{align*}
    \xi \cdot T &= T^{i}_{jk} \xi^{l}_i \e_l \otimes  \e^j \otimes \e^k \\
    T \cdot \xi &= T^{i}_{lk} \left[ \e_i \otimes ( \xi^j_l \cdot \e^l ) \otimes \e^k + \e_i \otimes  \e^j \otimes (\xi^k_l \cdot \e^l) \right] \\
        &= (T_{lk}^{i} \xi^l_j + T_{jl}^i \xi^{l}_k ) \e_i \otimes \e^j \otimes \e^k.
 \end{align*}
If we choose an arbitrary element $\alpha \in (\T^1_2)^* \equiv \T_1^2$ given by
\[
    \alpha = \alpha_i^{jk} \e^i \otimes \e_j \otimes \e_k
\]
we find that
\begin{align*}
    \langle \alpha , \xi \cdot T \rangle &= (\alpha_l^{jk} \xi^l_i) T_{jk}^i = (\alpha^{lk}_{i} T^{j}_{lk}) \xi^{i}_j\\
    \langle \alpha , T \cdot \xi \rangle &= ( \alpha^{lk}_{i} \xi^j_l + \alpha^{jl}_{i} \xi^{k}_{l}) T_{jk}^i = ( \alpha^{jk}_l T^l_{ik} + \alpha^{kj}_l T^l_{ki}) \xi^{i}_{j}.
\end{align*}
Therefore the heart operator is given by
\[
    \xi \heartsuit \alpha = ( \xi^l_i \alpha_l^{jk} -  \alpha^{lk}_{i} \xi^j_l - \alpha^{jl}_{i} \xi^{k}_{l}) \e^i \otimes \e_j \otimes \e_k
\]
and the diamond operator is
\[
    \alpha \diamondsuit T = ( \alpha^{jk}_l T^l_{ik} + \alpha^{kj}_l T^l_{ki} - \alpha^{lk}_{i} T^{j}_{lk} ) \e^i \otimes \e_j.
\]
Given a reduced Lagrangian $\ell: \mathfrak{gl}(n) \Join \T^1_2 \to \mathbb{R}$ we can denote $\mu = \frac{ \delta \ell}{\delta \xi}$ and $\gamma = \frac{ \delta \ell}{\delta T}$.
In terms of the basis $\e^i \otimes \e_j$ and $\e_i \otimes \e^j \otimes \e^k$ we may write the (right) Euler-Poincar\'e equations as:
\begin{align*}
\dot{\mu}^j_i  &=  \alpha^{lk}_{i} T^{j}_{lk} + \mu^j_k \xi^k_i - \mu_i^k \xi^j_k - \alpha^{jk}_l T^l_{ik} - \alpha^{kj}_l T^l_{ki}  \\
\dot{T}^i_{jk} &= \xi^l_i \alpha_l^{jk} -  \alpha^{lk}_{i} \xi^j_l - \alpha^{jl}_{i} \xi^{k}_{l}.
\end{align*}
  By restricting $\T^1_2$ to the subspace $\S^1_2$, we can obtain a Lie group which models second order jets of diffeomorphisms as demonstrated in proposition \ref{prop:2jets}. This example provides a first step towards the creation of higher-order, spatially accurate particle methods \cite[section 4]{JaRaDe2011}.  Moreover, the data of second order jets is necessary for the advection of quantities seen in complex fluids in which the advected parameters depend on gradients of the flow \cite{GayBalmaz2009,Holm2002}.  Therefore, the structures described here may prove useful in the construction of particle-based integrators for complex fluids as well.

\section{Conclusion}
In this paper, we have presented a variant of traditional semi-direct products, dubbed centered semi-direct products, and we have illustrated the associated Euler-Poincar\'{e} theory.  The diamond operator, the group multiplication, and the Lie bracket can all be seen as sums of the associated concepts for left and right semi-direct products.  As a result, the Euler-Poincar\'e theory associated with centered semi-direct products can also be seen as a sum of the left and right invariant Euler-Poincar\'{e} theories for semi-direct products.  Presently, many of these constructions remain fairly theoretical.  However, an isotropy group of the second order jet groupoid can be seen as a centered semi-direct product.  This has potential applications in simulation of complex fluids. We hope this paper provides a stepping stone towards realizing this application.

\bibliographystyle{amsplain}
\bibliography{CoJaMarsden_Memorial}

\end{document}